\documentclass[9pt,shortpaper,twoside,web]{ieeecolor}
\usepackage{generic}
\usepackage{cite}
\usepackage{amsmath,amssymb,amsfonts}
\usepackage{algorithmic}
\usepackage{graphicx}
\usepackage{textcomp}
\usepackage{color}
\input{color.lrt}
\usepackage{tikz}
\usetikzlibrary{calc,shapes}
\usetikzlibrary{positioning,arrows}
\usepackage{pgfplots}
\usetikzlibrary{fit,calc}
\pgfplotsset{compat=1.14}
\usepackage{psfrag}
\usepackage{graphicx}
\usepackage{amsmath,amsfonts,amsxtra}
\usepackage{mathrsfs}
\usepackage{todonotes}
\usepackage{mathtools,leftidx}
\usepackage{calligra}
\usepackage{float}
\usepackage{cleveref}
\usepackage[section]{placeins}
\crefname{equation}{}{}
\crefrangelabelformat{equation}{({#3}{#1}{#4}--{#5}{#2}{#6})}%
\crefrangelabelformat{subequation}{({#3}{#1}{#4}--{#5}\crefstripprefix{#1}{#2}#6)}%

\usepackage{environ}
\usepackage{bm} 
\usepackage{bbm}
\allowdisplaybreaks[1]

\newcommand{\col}[1]{\operatorname{col}(#1)}

\newcommand{\diag}[1]{\operatorname{diag}(#1)}

\newtheorem{assum}{Assumption}
\renewcommand{\d}{\mathrm{d}}
\renewcommand{\t}{^{\top}}

\newcommand{\tint}{\textstyle\int}

\newcommand{\e}{\text{e}}

\tikzset{
	state/.style={circle,draw,minimum size=6ex},
	arrow/.style={-latex, shorten >=1ex, shorten <=1ex}}

\DeclareRobustCommand\bluecircle{
	\tikz{
		\node (0,0){};
		\fill [blue](0,0.0001) circle (2.9pt);
		\fill [white](0,0.0001) circle (1.7pt);
	}
}
\DeclareRobustCommand\greencube{
	\tikz{
		\node[green, rectangle,thick,line width=0.5mm,minimum width=2mm,minimum height=2mm, draw,radius=2.5pt, rectangle, inner sep=0.6125pt, above right=0.001mm and -2.5mm] (0,0){};
	}
}\DeclareRobustCommand\redcross{
	$\mathbin{\tikz [x=1.2ex,y=1.2ex,line width=1pt, red] \draw (0,0) -- (1,1) (0,1) -- (1,0);}$
}

\newtheorem{thm}{Theorem}
\newtheorem{Lemma}{Lemma}
\newtheorem{rem}{Remark}

\def\BibTeX{{\rm B\kern-.05em{\sc i\kern-.025em b}\kern-.08em
    T\kern-.1667em\lower.7ex\hbox{E}\kern-.125emX}}
\begin{document}
\title{A Koopman Operator Approach to Data-Driven Control\\ of Semilinear Parabolic Systems}
\author{Joachim~Deutscher,~\IEEEmembership{Member,~IEEE}, Tarik~Enderes
\thanks{J. Deutscher and T. Enderes are with the Institute of Measurement, Control and Microtechnology, Ulm University, Germany. (e-mail: joachim.deutscher@uni-ulm.de; tarik.enderes@uni-ulm.de).
}
}
\maketitle

\begin{abstract}
This paper is concerned with the data-driven stabilization of unknown boundary controlled semilinear parabolic systems. The nonlinear dynamics of the system are lifted using a finite number of eigenfunctionals of the Koopman operator related to the autonomous semilinear PDE. This results in a novel data-driven finite-dimensional model of the lifted dynamics, which is amenable to apply design procedures for finite-dimensional systems to stabilize the semilinear parabolic system. In order to facilitate this, a bilinearization of the lifted dynamics is considered and feedback linearization is applied for the data-driven stabilization of the semilinear parabolic PDE. This reveals a novel connection between the assignment of eigenfunctionals to the closed-loop Koopman operator and feedback linearization. By making use of a modal representation, exponential stability of the closed-loop system in the presence of errors resulting from the data-driven computation of eigenfunctionals and the bilinearization is verified. The data-driven controller directly follows from applying generalized eDMD to state data available for the semilinear parabolic PDE. An example of an unstable semilinear reaction-diffusion system with finite-time blow up demonstrates the novel data-driven stabilization approach.
\end{abstract}


\section{Introduction}

\subsubsection{Background and Motivation} The stabilization of boundary controlled semilinear parabolic PDE systems arises in many applications. For example, thermal chemical-reacting and plasma systems or crystal growth processes to name but a few. If the controller has to be designed in a \emph{late-lumping} fashion, i.e., directly taking the nonlinear PDE system for the design into account, several methods are available in the literature. This includes nonlinear backstepping designs (see \cite{Vaz08,Vaz08a}), optimal control (see \cite{Kre22}) or discontinuous control (see \cite{Orl02}). Another wildely applied approach is to determine an approximation of the semilinear parabolic system and then design a controller for the resulting finite-dimensional model, i.e., applying \emph{early lumping}. Prominent methods are the Galerkin approach (see \cite{Cr01} for semilinear parabolic systems with distributed input) and the proper orthogonal decomposition (POD) (see \cite[Ch. 12]{Br22}). All these methods, however, require the knowledge of the PDE for the design. 

Recently, data-driven control methods have become more and more important due to the increasing availability of data. They allow to extend the controller design to complex system classes, where the first principle modeling is very tedious (see \cite{Br22}). In addition, no time consuming system modeling and parameter identification are necessary. In this respect, the \emph{Koopman operator approach} has emerged as a ubiquitous tool to deal with the data-driven analysis and synthesis of nonlinear dynamical systems (see \cite{Ma20}). This new viewpoint lifts the nonlinear autonomous system dynamics on an infinite-dimensional space of observables, where the evolution of the latter can be described by the linear Koopman operator globally. Hence, linear system and operator theory can be applied for the analysis of the underlying nonlinear dynamics. In addition, systematic methods exist for the data-driven approximation of the Koopman operator with guaranteed finite-data bounds (see \cite{Sch23}) in form of the linear regression-based \emph{extended dynamic mode decomposition (eDMD)} (see \cite{Ku16,Br22}). Currently, the focus of the Koopman approach for data-driven control has been on finite-dimensional nonlinear systems (see \cite{Ma20} for examples), while only few results are available for nonlinear PDE systems. In particular, the contributions \cite{Pe19,Ar18} apply the Koopman theory to nonlinear PDE systems in order to solve data-driven optimal control problems. Therein, however, the approximation of the Koopman operator is based on a spatial discretization of the underlying PDE dynamics. In \cite{Na20} it has been shown, that this is not necessary by introducing functionals as observables. This allows a direct spectral analysis of the Koopman operator and thus of the nonlinear PDE system. First applications of these new ideas include the solution of identification problems in \cite{Mau21} and the Koopman modal control of linear parabolic systems in \cite{Deu24}. So far, however, the results in \cite{Na20} have not been used for the data-driven control of nonlinear PDE systems, which is the obvious application of the Koopman theory.

\subsubsection{Contribution} This paper considers the data-driven state feedback control for unknown semilinear parabolic systems, for which only state data is available. For this, the underlying nonlinear dynamics are lifted by using a finite number of open-loop eigenfunctionals of the Koopman operator for the semilinear parabolic system. As a result, a Koopman modal representation in form of a finite-dimensional nonlinear system with linear eigendynamics and a state-dependent input vector field is obtained, which globally describes the system dynamics. Hence, the dominant dynamics are captured with a lower order model, when compared to a Galerkin-type approximation using the eigenfunctions related only to the linear part of the semilinear parabolic PDE (see \cite[Ch. 4]{Cr01}). The new Koopman lifting extends previous results in \cite{Gos21,Kai21,Sur16} for finite-dimensional nonlinear systems to infinite dimensions. For finite-dimensional systems \cite{Kai21} directly uses the Koopman model with state dependent input vector field for the design, while the results in \cite{Gos21} rely on a subsequent bilinearization for the controller and observer design. In the paper 
the \emph{optimal bilinearization method} in \cite{Gos21} is generalized to the semilinear parabolic systems in question. 
The new data-driven bilinear system description provides a firm basis to apply various design methods for finite-dimensional nonlinear systems to stabilize unknown semilinear parabolic PDE systems. In the paper, feedback linearization is utilized to facilitate the design of a stabilizing controller assigning a finite number of closed-loop eigenvalues. In particular, the \emph{single-step approach} proposed in \cite{Ka00} is applied, which circumvents the restrictive involutivity conditions necessary to map the system into nonlinear controller form (see \cite[Th. 4.2.3]{Is95}). Hence, feedback linearization under weak conditions and thus for a large class of semilinear parabolic PDE systems is possible. More importantly, it is shown that if the Koopman modal representation can be reformulated with an input vector field depending only on the open-loop eigenfunctionals without error, then this feedback linearizing controller also assigns a set of desired eigenfunctionals to the closed-loop system. This extends the results in \cite{Deu24} for linear parabolic systems to the semilinear case. This pleasing result directly exploits the inherent property of the Koopman operator that the system dynamics have a nice linear representation in the coordinates of the eigenfunctionals and relates it to feedback linearization. It is demonstrated that this design method can be reformulated as a nonlinear change of coordinates using modal coordinates w.r.t. the eigenfunctions of the linear part in the semilinear parabolic PDE. This allows to systematically verify closed-loop stability also in the presence of the errors appearing from the data-driven computation of the open-loop eigenfunctionals and in the system simplification. The design only requires a finite number of Koopman eigenvalues and the corresponding eigenfunctionals. They are determined from the available state data by making use of eDMD, which is generalized in \cite{Mau21} to PDE systems using functionals as observables. This results in a general and new data-driven approach for the stabilization of semilinear parabolic systems systematically exploiting the spectral properties of the Koopman operator for nonlinear PDEs.

\subsubsection{Organization} The next section introduces the considered data-driven stabilization problem. Then, Sec. \ref{sec:feedoefktl} presents the derivation of the lifted dynamics and the corresponding bilinearization. The design of a stabilizing controller using feedback linearization is dealt with in Sec. \ref{sec:feedoefktl2}. The data-driven solution of the Koopman eigenvalue problem using state data on the basis of eDMD is contained in Sec. \ref{sec:dewp}. Finally, the data-driven stabilization of a semilinear diffusion-reaction system illustrates the results of the paper in Sec. \ref{sec:ex}. 


\section{Problem formulation}\label{sec:prob}
\subsubsection{Similinear Parabolic System} Consider the \emph{unknown semilinear parabolic system}
\begin{subequations}\label{plant}
\begin{align}
 \dot{x}(z,t) &= \rho x''(z,t) + a(z)x(z,t) + f(z,x(z,t))  &&\label{parasys}\\
      x'(0,t) &= 0,                                          &&\hspace{-2cm} t > 0  \label{parasys1}\\ 
      x'(1,t) &= u(t),                                       &&\hspace{-2cm} t > 0, \label{parasys2}
\end{align}
\end{subequations}
with the state $x(z,t) \in \mathbb{R}$ defined on $(z,t) \in [0,1]\times\mathbb{R}^+$ and the input $u(t) \in \mathbb{R}$. It is assumed that $\rho > 0$ and $a \in C[0,1]$. Moreover, $f$ is in $C[0,1]$ w.r.t. $z$ and is analytic w.r.t. $x$ as well as $f(z,0) = f'(z,0) = 0$, $z \in [0,1]$, such that  $x = 0$ is an \emph{equilibrium} of \eqref{plant}. The initial conditions (ICs) of \eqref{plant} are $x(z,0) = x_0(z) \in \mathbb{R}$. 

It is assumed that \emph{state data} $x(z,t_i) \in \mathbb{R}$ at sampling times $t_i \geq 0$, $i = 1,\ldots, M$, is collected for an initial value $x(z,0) = x_0(z) \neq 0$ and $u \equiv 0$. For estimating $\rho$ from data only \emph{one} time derivative of the state $\dot{x}(z,t_0) \in \mathbb{R}$ for some $t_0 \geq 0$ is required, which results for a constant input $u(t) = u_0 \in \mathbb{R}$, $t \geq 0$.

\subsubsection{Koopman Operator and Generator} In order to introduce the Koopman operator, define the state $x(t) = x(\cdot,t)$ and consider the initial boundary value problem \eqref{plant}  for $u = 0$, which can be represented by the abstract initial value problem (IVP)
\begin{equation}\label{absivp1}
 \dot{x}(t) = \mathcal{A}x(t)  + f(\cdot,x(t)) = \mathcal{F}x(t), \quad x(0) = x_0 \in \mathbb{X}, t > 0,
\end{equation}
in the space $\mathbb{X} = L_2(0,1)$ of complex $L_2$-functions endowed with the inner product $\langle x_1, x_2\rangle = \int_0^1x_1(\zeta)\overline{x_2(\zeta)}\d\zeta$. Therein, the linear operator $\mathcal{A}h = \rho h'' + ah$, $h \in D(\mathcal{A}) \subset \mathbb{X}$ with $h \in D(\mathcal{A}) = \{h \in H^2(0,1)\,|\, h'(0) = h'(1) = 0\}$ as well as the nonlinear operator $\mathcal{F} : D(\mathcal{F}) \subset \mathbb{X} \to \mathbb{X}$ with $\mathcal{F}h = \mathcal{A}h + f(\cdot,h)$ and $D(\mathcal{F}) = D(\mathcal{A})$ defining the right hand side in \eqref{absivp1} are used. The following assumption is required in order to ensure well-posedness of the IVP \eqref{absivp1}.
\begin{assum}[Wellposedness] \hfill The nonlinearity 
	 $f : [0,1] \times \mathbb{X} \to \mathbb{X}$ in \eqref{absivp1} is \emph{locally Lipschitz continuous} in $x$, i.e., for a constant $r > 0$ there exists a positive constant $C(r)$ such that  
	\begin{equation}\label{Lcond}
	\|f(\cdot,x_1) - f(\cdot,x_2)\| \leq C \|x_1 - x_2\|
	\end{equation}
	holds for all $x_1, x_2 \in \mathbb{X}$ with $\|x_i\| \leq r$, $i = 1,2$. 
\end{assum}
This assumption ensures existence of a unique \emph{mild solution} for \eqref{absivp1} on $t \in [0,t_{\text {max}})$ with some $0 < t_{\text{max}} \leq \infty$ (see \cite[Th. 1.4, Ch. 6]{Paz83}). If the Lipschitz constant in \eqref{Lcond} is independent of $r$, i.e., $f$ is \emph{globally Lipschitz} continuous in $x$, then the mild solution is globally valid for all $t \geq 0$ (see \cite[Th. 1.2, Ch. 6]{Paz83}). Moreover, if $x_0 \in D(\mathcal{F})$, then the mild solution is a \emph{classical solution} (see \cite[Th. 1.5, Ch. 6]{Paz83}). Hence, there exists a \emph{semiflow} $\mathcal{T}_{\mathcal{F}}(t) : \mathbb{X} \to \mathbb{X}$, $t \geq 0$, generated by $\mathcal{F}$ describing the time evolution of the state $x$. The dynamics of \eqref{plant} are described by lifting the state $x$ using the \emph{observables} $\theta : D(\theta) = D(\mathcal{F}) \to \mathbb{C}$ being elements of the infinite-dimensional space $\mathbb{O}$ of observables. This is the space of bounded continuous complex-valued functionals endowed with the supremum norm $\|\theta\|_{\infty} = \sup_{x \in D(\theta)}|\theta[x]|$. Therein, the domain $D(\theta) = D(\mathcal{F})$ of the functionals in $\mathbb{O}$ is invariant under the semiflow $\mathcal{T}_{\mathcal{F}}(t)$ for $x_0 \in D(\mathcal{F})$ so that they are well-defined for all $x \in D(\theta)$. 
Since any measurement for \eqref{plant} is a functional, the observables $\theta[x]$ have the intuitive meaning as generalized measurements for \eqref{plant}. Then, the \emph{Koopman operator} $\mathscr{K}(t): \mathbb{O} \to \mathbb{O}$  associated with the semiflow $\mathcal{T}_{\mathcal{F}}(t)$ of \eqref{absivp1} can be introduced by the composition
\begin{equation}\label{Kop}
\mathscr{K}(t)\theta[x] = \theta[\mathcal{T}_{\mathcal{F}}(t)x], \quad \theta \in \mathbb{O}, x \in  D(\theta) = D(\mathcal{F}), t \geq 0,
\end{equation} 
which describes the time evolution of the observables. For each $t \in \mathbb{R}^+$ the map $\mathscr{K}(t)$ is an infinite-dimensional linear operator (see, e.g., \cite{Na20}). This is one of the crucial properties of the Koopman operator, which allows to analyze the uncontrolled semilinear parabolic system \eqref{plant} using linear infinite-dimensional system theory in the space of observables. 

The \emph{Koopman generator} $\mathscr{A} : D(\mathscr{A}) \subset \mathbb{O} \to \mathbb{O}$ is defined by the \emph{Lie generator} 
\begin{equation}\label{Kgen}
\mathscr{A}\theta[x] = \lim_{t \to 0^+}\frac{\mathscr{K}(t)\theta[x] - \theta[x]}{t} = \delta \theta[x;\mathcal{F}x]
\end{equation}
for all $x \in  D(\theta)$ associated to the Koopman semigroup $\mathscr{K}(t)$ (see \cite{Mau21,Do96}), in which the limit has to be taken in the topology of $\mathbb{O}$. Then, the domain $D(\mathscr{A})$ is the set of elements $\theta[x] \in \mathbb{O}$, for which the limit in \eqref{Kgen} exists. For determining the Koopman generator $\mathscr{A}$ the \emph{G\^ateaux differential} $\delta \theta[x;\mathcal{F}x]$ of the observable $\theta[x]$ in the direction of the function $\mathcal{F}x$ can be used. In particular, the G\^ateaux differential reads
\begin{align}\label{fablcalc}
		\mathscr{A}\theta[x] &= \delta \theta[x;\mathcal{F}x] \!=\! \d_{\epsilon}\theta[x + \epsilon\mathcal{F}x]|_{\epsilon = 0}\nonumber\\
		                &= \int_0^1(\delta_x\theta[x])(\zeta)\mathcal{F}x(\zeta)\d\zeta = \mathcal{L}_{\mathcal{F}}\theta[x],
\end{align}
which has the representation in form of a linear integral operator resulting in a functional. Therein, $\delta_x\theta[x] = \delta \theta[x]/\delta x$ denotes the \emph{functional derivative} of $\theta[x]$ w.r.t. the function $x$, which yields a function. Hence, \eqref{fablcalc} can be seen as a generalization $\mathcal{L}_{\mathcal{F}}$ of the Lie-derivative known from finite-dimensions. With this, the \emph{evolution equation} for the observables $\theta[x]$, i.e., $\theta[x(t)] = \mathscr{K}(t)\theta[x(0)]$, $t \geq 0$, takes the form
\begin{equation}\label{timeevobsp}
 \d_t\theta[x(t)] = \mathscr{A}\theta[x(t)], \quad t > 0, \theta[x(0)] \in D(\mathscr{A}) \subset \mathbb{O}.
\end{equation}
For further details concerning the definition and existence of the Koopman operator and generator for infinite-dimensional systems, the interested reader is referred to \cite{Na20,Mau21}.


\subsubsection{Koopman Eigenvalue Problem} In what follows, the eigenvalue problem for the Koopman generator related to PDEs is shortly reviewed (see \cite{Na20} for more details). For this, the  following assumption is required.
\begin{assum}[Conjugacy]\label{ass:conj}
 The system \eqref{absivp1} is locally \emph{conjugate} to its linearization around the equilibrium $x = 0$.
\end{assum}
This means that there locally exists a diffeomorphism $\tilde{x} = \Phi(x)$ such that $\Phi(\mathcal{T}_{\mathcal{F}}(t)x_0) = \mathcal{T}_{\mathcal{A}}(t)\Phi(x_0)$, $t \geq 0$, holds, in which $\mathcal{T}_{\mathcal{F}}(t)$ is the semiflow generated by the nonlinear operator $\mathcal{F}$ and $\mathcal{T}_{\mathcal{A}}(t)$ is the $C_0$-semigroup generated by the linear operator $\mathcal{A}$, i.e., the linearization of $\mathcal{F}$. Consequently, the semilinear parabolic system can be represented locally by its linearization so that the spectral properties of the latter described by the linear operator $\mathcal{A}$ are inherited by the Koopman operator $\mathscr{K}(t)$ (see \cite{Na20}). In particular, the linear operator  $-\mathcal{A}$ is a self-adjoint \emph{Sturm-Liouville operator} with a pure discrete point spectrum consisting of real simple eigenvalues $\lambda_i$, $i \in \mathbb{N}$ (see \cite{Del03}). If the equilibrium $x = 0$ of \eqref{absivp1} is isolated and \emph{hyperbolic}, then the existence of a local conjugacy is ensured by a generalization of the Hartman-Grobman theorem (see \cite{Lu91}).
With this, the \emph{eigenfunctionals} $\varphi_i[x] \in \mathbb{O}$, $i \in \mathbb{N}$, w.r.t. the \emph{Koopman eigenvalues} $\lambda_i$ (\emph{K-eigenvalues} for short) are the nontrivial solutions of the \emph{Koopman eigenvalue problem}
\begin{equation}\label{Kopewp}
	\mathscr{K}(t)\varphi_i[x] = \e^{\lambda_it}\varphi_i[x], \quad i \in \mathbb{N},
\end{equation}
for the Koopman operator $\mathscr{K}(t)$. For smooth eigenfunctionals $\varphi_i[x]$ it is verified in \cite{Na20} that they also solve the \emph{eigenvalue problem}
\begin{equation}\label{kewpp}
\mathscr{A}\varphi_i[x] = \mathcal{L}_{\mathcal{F}}\varphi_i[x] = \lambda_i\varphi_i[x],
\end{equation}
$\varphi_i[x] \in D(\mathscr{A})$, $i \in \mathbb{N}$, for the Koopman generator $\mathscr{A}$. 
\begin{rem}\label{rem:prineig}
It should be noted that \eqref{Kopewp} and \eqref{kewpp} only yield the \emph{principal K-eigenvalues} (see \cite{Na20}), which coincide with the eigenvalues of $\mathcal{A}$ and thus determine the stability of the parabolic system \eqref{plant}. Furthermore, also $\varphi^{\mu_1}_1[x]\cdot\ldots\cdot\varphi^{\mu_n}_n[x]$ with $\mu_i \in \mathbb{N}_0$ satisfying $\mu_1 + \ldots + \mu_n > 0$, $n > 1$, are additional eigenfunctionals of the Koopman generator w.r.t. the K-eigenvalues $\sum_{j=1}^n\mu_j\lambda_j$. \hfill $\triangleleft$	 
\end{rem}

\begin{rem}
For semilinear parabolic PDEs, the eigenfunctionals $\varphi_i[x]$ can, in general, not be determined explicitly. Closed-form solutions are only available for PDEs, where an explicit conjugacy to a linear PDE is known. This concerns more general types of semilinear  parabolic PDEs such as the Burgers equation, the viscous Hamilton-Jacobi PDEs and the nonlinear phase-diffusion equation, where an explicit conjugacy is given by Hopf-Cole-type transforms (see \cite{Ku18,Na20}). 
	\hfill $\triangleleft$	 
\end{rem}

\subsubsection{Data-driven Stabilization Problem} 
Consider the local \emph{nonlinear transformation} 
\begin{equation}\label{ntraf}
 \tilde{w} = \Phi(w) \in \mathbb{R}^n
\end{equation}
 with a local diffeomorphism $\Phi : \mathbb{R}^n \to \mathbb{R}^n$ around the origin, $w = \bm{\varphi_n}[x]$ and the open-loop eigenfunctionals $\bm{\varphi_n}[x(t)] = \col{\varphi_1[x(t)],\ldots,\varphi_n[x(t)]} \in \mathbb{C}^n$. Then, the \emph{state feedback controller}
\begin{equation}\label{sf}
 u = \alpha(w) -k\t  \Phi(w) 
\end{equation}
with $\alpha$ analytic around the origin and $\alpha(0) = 0$, $\partial_w\alpha|_{w = 0} = 0$ as well as the feedback gain $k \in \mathbb{R}^n$ can be introduced. With this, the closed-loop system takes the form
\begin{subequations}\label{kloop}
	\begin{align}
	  \dot{x}(z,t) &= \rho x''(z,t) + a(z)x(z,t) + f(z,x(z,t))\\
 		   x'(0,t) &= 0\\
		   x'(1,t) &= \alpha(w(t)) -k\t  \Phi(w(t)).
	\end{align}	
\end{subequations}
For determining \eqref{ntraf} and \eqref{sf}, the closed-loop system \eqref{kloop}  is considered in the modal coordinates $x_i = \langle x,\phi_i\rangle$, $i \in \mathbb{N}$. Therein, 
$\phi_i$, $i \in \mathbb{N}$, are the eigenvectors of $\mathcal{A}$ w.r.t. the eigenvalues $\lambda_i$ following from the eigenvalue problem
$\mathcal{A}\phi_i = \lambda_i\phi_i$, $\phi_i \in D(\mathcal{A})$. Since $-\mathcal{A}$ is a Sturm-Liouville operator, the sequence $\phi_i$, $i \in \mathbb{N}$, is an orthonormal Riesz basis for $\mathbb{X}$ (see \cite{Del03}). Then, with $x = \sum_{i=1}^\infty x_i\phi_i$ the transformation
\begin{subequations}\label{modtraf}
\begin{align}
 \tilde{w} &= \Phi(w) =  \Phi\Big(\bm{\varphi_n}[\sum_{i=1}^\infty x_i\phi_i]\Big) =  \bar{\Phi}(x_1,x_2,\ldots) \in \mathbb{R}^n\label{wtiltraf}\\
       x_i &= \langle x,\phi_i\rangle \in \mathbb{R}, \quad i > n,
\end{align}
\end{subequations}
can be applied to map the closed-loop system \eqref{kloop} represented in the modal coordinates $x_i = \langle x,\phi_i\rangle$, $i \in \mathbb{N}$, into 
\begin{subequations}\label{cl}
	\begin{align}
		\dot{\tilde{w}}(t) &= \tilde{A}\tilde{w}(t) + e_{\tilde{w}}(\tilde{w}(t),\bm{x_\infty}(t)) && \hspace{-0.95cm}\label{finlin}\\
		\dot{x}_i(t) &= \lambda_ix_i(t) + f_i(\tilde{w}(t),\bm{x_\infty}(t)) - \phi_i(1)k\t \tilde{w}(t) &&  \hspace{-0cm} \label{infsub}
	\end{align}	 
\end{subequations}
for $i > n$ with $\lambda_i < 0$ and $f_i = \langle f(\cdot,x),\phi_i\rangle$. Therein,  $\bm{x_\infty} = (x_{n+1},x_{n+2},\ldots) = (\langle x,\phi_{n+1}\rangle, \langle x,\phi_{n+2}\rangle, \ldots )\in \ell^2$ holds, in which the space $\ell_2 = \{(h_1,h_2,\ldots)\,|\, h_i \in \mathbb{C}, \sum_{i=1}^{\infty}|h_i|^2 < \infty\}$ is the Hilbert space of absolutely square summable sequences $(h_i)_{i\in\mathbb{N}}$ composed of scalars $h_i \in \mathbb{C}$ equipped with the inner product $\langle x,y \rangle_{\ell_2} = \sum_{i=1}^{\infty}x_i\overline{y_i}$. The transformation \eqref{ntraf} and the state feedback in \eqref{sf} have to be determined such that $\tilde{A} \in \mathbb{R}^{n \times n}$ is a Hurwitz matrix and the origin of the resulting closed-loop system \eqref{cl} is exponentially stable with the closed-loop eigenvalues of the linearization given by $\sigma(\tilde{A})$ and $\lambda_i$, $i > n$. Since the system \eqref{plant} is unknown, a \emph{data-driven approach} for this stabilization problem is presented by determining a finite number of open-loop eigenfunctionals from the available data using \emph{extended dynamic mode decomposition (eDMD)}.

\section{Koopman Modal System Representation}\label{sec:feedoefktl}
For determining \eqref{ntraf} and \eqref{sf}, a finite-dimensional system representation of the dominant dynamics to be modified by state feedback is derived in \emph{Koopman modal coordinates} $\varphi_i[x]$, $i = 1,\ldots,n$. They can be obtained from the \emph{extended dynamic mode decomposition (eDMD)} using available data (see \cite{Mau21,Ku16} and Section \ref{sec:dewp}). Consider the evolution equation 
\begin{subequations}\label{vewp}
	\begin{align}
		\d_t\varphi_i[x(t)] &= \int_0^1(\delta_x\varphi_i[x(t)])(\zeta)(\rho x''(\zeta,t) + a(\zeta)x(\zeta,t)\nonumber\\
		&\quad  + f(\zeta,x(\zeta,t)))\d\zeta\label{evphi}\\
		x'(0,t) &= 0    \label{bcu1} \\
		x'(1,t) &= u(t) \label{bcu2}
	\end{align}	 
\end{subequations}
for the open-loop eigenfunctionals $\varphi_i[x]$, $i = 1,\ldots,n$ (cf. \eqref{timeevobsp}, \eqref{kewpp}) and taking the input into account. Therein, it is assumed that $\varphi_i[x]$ is smooth so that the functional derivative in \eqref{evphi} exists. In what follows, the orthonormal basis spanned by the eigenvectors $\phi_i$, $i \in \mathbb{N}$, w.r.t. $\mathcal{A}$ for $\mathbb{X}$ will be used. Here, the fact that $\mathcal{A}$ is unknown leads to no problems, because the  $\phi_i$, $i \in \mathbb{N}$, are only needed for the analysis but not for the design of the controller.
Hence, the expansion $x = \sum_{i=1}^\infty\langle x,\phi_i\rangle\phi_i = \sum_{i=1}^\infty x_i\phi_i \in D(\theta) \subset \mathbb{X}$, can be inserted in $\varphi_i[x]$ yielding the representation
\begin{equation}\label{eigivar}
	\varphi_i(\bm{x}) = \varphi_i[\sum_{j=1}^\infty x_j\phi_j]
\end{equation}
of the open-loop eigenfunctionals $\varphi_i[x]$ as a function of an infinite countable number of variables $\bm{x} = (x_1,x_2,\ldots) \in \ell^2$. It is shown in \cite{Ven18,Ven21} that the functional derivative of \eqref{eigivar} becomes
\begin{equation}\label{funcdiff}
	\delta_x\varphi_i[x] = \sum_{j=1}^\infty\partial_{x_j}\varphi_i(\bm{x})\phi_j
\end{equation}
with $\partial_{x_j}\varphi_i(\bm{x}) = \langle \delta_x\varphi_i[x],\phi_j \rangle$. Then, the evolution equation \eqref{vewp} takes the form
\begin{multline}\label{infpde}
	\d_t\varphi_i(\bm{x}(t)) = \sum_{j=1}^\infty\partial_{x_j}\varphi_i(\bm{x}(t))\int_0^1\phi_j(\zeta)(\rho x''(\zeta,t)\\ + a(\zeta)x(\zeta,t)
	+ f(\zeta,x(\zeta,t)))\d\zeta.
\end{multline}
Two integrations by parts yield $\tint_0^1\rho\phi_j(\zeta)x''(\zeta,t)\d\zeta = \rho\phi_j(1)u(t) + \tint_0^1\rho\phi''_j(\zeta)x(\zeta,t)\d\zeta$ using the BCs \eqref{parasys1}, \eqref{parasys2} to be verified by the eigenvector $\phi_j$ and the BCs \eqref{bcu1}, \eqref{bcu2}. Then, the eigenvalue problem $\rho\phi''_j + a\phi_j = \lambda_j\phi_j$, $j \in \mathbb{N}$, for $\mathcal{A}$  shows that $\d_t\varphi_i(\bm{x}(t)) = \sum_{j=1}^\infty\partial_{x_j}\varphi_i(\bm{x}(t))\int_0^1(\lambda_j\phi_j(\zeta)x(\zeta,t) + f(\zeta,x(\zeta,t)))\d\zeta  + g_i[x(t)]u(t)$ with
\begin{equation}\label{gdef}
g_i[x] = \rho\sum_{j=1}^\infty\partial_{x_j}\varphi_i(\bm{x})\phi_j(1) = \rho(\delta_x\varphi_i[x])(1)
\end{equation}
implied by \eqref{funcdiff}. Applying the same reasoning to the G\^ateaux differential $\mathcal{L}_{\mathcal{F}}$ in \eqref{fablcalc} finally leads to
\begin{equation}\label{cfunc2}
	\d_t\varphi_i[x(t)] =
	\mathcal{L}_{\mathcal{F}}\varphi_i[x(t)]
    + g_i[x(t)]u(t).
\end{equation}
Note that \eqref{cfunc2} is also the evolution equation valid for any observable $\theta[x] \in \mathbb{O}$ in the presence of an input $u$. Hence, this result can also be useful if \emph{generator eDMD} is applied to determine a finite-dimensional continuous-time surrogate model (see \cite{Klus20}). With \eqref{kewpp}, the \emph{Koopman modal system representation}
\begin{equation}\label{nsys}
 \d_t\bm{\varphi_n}[x(t)] = \Lambda_n\bm{\varphi_n}[x(t)] + \rho(\delta_x\bm{\varphi_n}[x(t)])(1)u(t)
\end{equation}
follows with $\bm{\varphi_n}[x] = \col{\varphi_1[x],\ldots,\varphi_n[x]}$ and $\Lambda_n = \diag{\lambda_1,\ldots,\lambda_n}$. Note that \eqref{nsys} is a finite-dimensional lifted nonlinear dynamics \eqref{plant} in the open-loop eigenfunctionals, which is defined globally in the domain of the eigenfunctionals. This generalizes the corresponding results in \cite{Gos21,Ia24} for finite-dimensional nonlinear systems. As anticipated, a nonlinear system with linear diagonal eigendynamics and state dependent input vector field is also obtained for the semilinear parabolic system \eqref{plant}. The former property means that the eigenfunctionals diagonalize the nonlinear operator $\mathcal{F}$ (see \eqref{absivp1}), which extends the modal analysis of linear parabolic systems using eigenfunctions to the nonlinear case. An approach to directly deal with this type of systems is presented in \cite{Kai21}, which requires to solve a state-dependent Riccati equation online. Alternatively, the Koopman modal representation \eqref{nsys} can be viewed as a special structured \emph{linear parameter-varying system}, which allows to apply convex optimization methods for the controller design (see \cite{Ia24}).
In this paper, the system \eqref{nsys} is bilinearized so that finite-dimensional nonlinear control design procedures can be directly used. In particular, \eqref{nsys} is rewritten in the form of the \emph{bilinear system}
\begin{equation}\label{nsys2}
	\d_t\bm{\varphi_n}[x(t)] = \Lambda_n\bm{\varphi_n}[x(t)] + (b + N\bm{\varphi_n}[x(t)] +  e_u[x(t)])u(t) 
\end{equation}
with $b \in \mathbb{R}^n$ and $N \in \mathbb{R}^{n\times n}$ subject to the error $e_u[x(t)] =  \rho(\delta_x\bm{\varphi_n}[x(t)])(1) - (b + N\bm{\varphi_n}[x(t)]) \in \mathbb{R}^n$. The latter has then to be taken into account in the design. Only in rare cases it will be possible to choose $b$, $N$ and $n$ so that $e_u[x] = 0$, i.e., if the system \eqref{nsys} is \emph{finitely bilinearizable} (see \cite{Gos21} for a related result concerning finite-dimensional nonlinear systems). Therefore, the vector $b$ and the matrix $N$ can only be determined to minimize $e_u[x]$.  For this, the approach in \cite{Gos21} is extended to the semilinear parabolic systems \eqref{plant}.  This requires to introduce the \emph{cylinder functional} $\varphi_i(\bm{\xi_m}) = \varphi_i[\mathcal{P}_mx]$, $\mathcal{P}_mx \in D(\theta)$, in which  $\mathcal{P}_m$ is the \emph{projection operator} $\mathcal{P}_mh = \sum_{i=1}^m\langle h,\vartheta_i\rangle\vartheta_i$, $h \in \mathbb{X}$,
with $\vartheta_i \in D(\mathcal{A})$, $i \in \mathbb{N}$, any known orthonormal basis for $\mathbb{X}$ and
$\bm{\xi_m} = (\xi_1,\ldots,\xi_m) = (\langle x,\vartheta_1\rangle,\ldots,\langle x,\vartheta_m\rangle) \in \mathbb{C}^m$. It is shown in \cite{Ven21} that $\varphi_i(\bm{\xi_m})$ converges uniformly to  $\varphi_i[x]$ for $m \to \infty$ in any compact subset of $\mathbb{X}$. Hence, the cylinder functional qualifies as an approximation for \eqref{eigivar}. 
Consider $g(\bm{\xi_m}) = \rho\delta_x\bm{\varphi_n}(\bm{\xi_m})(1)$, which follows from \eqref{gdef} with $g = \col{g_1,\ldots,g_n}$, $m \geq n$, on $\mathbb{R}^m$, i.e., using the corresponding cylinder functional. With the inner product $\langle x,y \rangle_{\mathbb{R}^m} = \tint_{\mathbb{R}^m}x\t(\bm{\xi_m})y(\bm{\xi_m})\d\bm{\xi_m}$ inducing the norm\linebreak $\|\cdot\|_{\mathbb{R}^m} = (\langle \cdot,\cdot \rangle_{\mathbb{R}^m})^{1/2}$ on $\mathbb{R}^m$, the optimization problem
\begin{equation}\label{minprob}
 \min_{b,N}\|g -  
 \begin{bmatrix}
 b & N
 \end{bmatrix}
 \bm{\bar{\varphi}_n}\|_{\mathbb{R}^m}
\end{equation}
with $\bm{\bar{\varphi}_n}(\bm{\xi_m}) = \col{1,\bm{\varphi_n}(\bm{\xi_m})}$ can be formulated. The following theorem shows that this minimization problem has an analytic solution.
\begin{thm}[Optimal bilinearization]\label{lem:optbilin}
Define the \emph{Gramian} $G = \langle\bm{\bar{\varphi}_n} ,\bm{\bar{\varphi}\t_n}\rangle_{\mathbb{R}^m} \in \mathbb{R}^{(n+1) \times (n+1)}$ and $R = \langle g,\bm{\bar{\varphi}\t_n}\rangle_{\mathbb{R}^m} \in \mathbb{R}^{n \times (n+1)}$. Then, the minimization problem \eqref{minprob} has the solution $[b, N] = RG^\dagger$, in which $G^{\dagger}$ is the \emph{Moore-Penrose inverse} (see, e.g., \cite[Ch. 12.8]{La85}). If the set of functions $1,\varphi_1(\bm{\xi_m}),\ldots,\varphi_n(\bm{\xi_m})$ is linearly independent on $\mathbb{R}^m$, i.e., $\det G \neq 0$, then a unique solution $[b,N] = RG^{-1}$ is obtained.
\end{thm}
The proof of this theorem directly follows from \cite{Gos21}. For large $m$ it may be cumbersome to compute the optimal solution, because this requires to evaluate $m$-dimensional integrals. A possible remedy is to use Monte Carlo methods for high dimensional integration (see \cite{Tan24}). In addition, a simple data-driven method for determining the optimal $b$ and $N$ numerically is provided in Sec. \ref{sec:dewp}. 

The resulting bilinear system may be the starting point to extend the optimal control approach in \cite{Gos21} or the Lyapunov control method in \cite{Huan18} to infinite dimensions. Recently, robust control was applied to the bilinear approximation model for nonlinear finite-dimensional systems in \cite{St24} guaranteeing the stability for finite state data.  
In the sequel, however, a feedback linearization approach is used for \eqref{nsys2} and is related to the closed-loop eigenfunctionals.

\section{Feedback Linearization}\label{sec:feedoefktl2}
By introducing the state $w = \bm{\varphi_n}[x]$ one obtains 
\begin{equation}\label{bilin}
 \dot{w}(t) = \Lambda_nw(t) + e_x[x(t)] + (b + Nw(t) + e_u[x(t)])u(t)
\end{equation}
for \eqref{nsys2} with the additional error $e_x[x]$, that also takes errors in the determination of the eigenfunctionals into account. Consider for \eqref{bilin} the local nonlinear transformation $\Phi : \mathbb{R}^n \to \mathbb{R}^n$ in \eqref{ntraf} and the state feedback \eqref{sf}.
Then, \eqref{bilin} is mapped into $\dot{\tilde{w}}(t) = \tilde{A}\tilde{w}(t) + e_{\tilde{w}}(\tilde{w}(t),x(t))$ with the error
\begin{equation}\label{etilde}
  e_{\tilde{w}}(\tilde{w},x) = \partial_w\Phi\circ\Phi^{-1}(\tilde{w})(e_x[x] + e_u[x](\alpha(\Phi^{-1}(\tilde{w}))-k\t\tilde{w}))
\end{equation}
if $\Phi(w)$ satisfies the first-order quasilinear PDE
\begin{equation}\label{spde21}
 \partial_w\Phi(w)(\Lambda_n w + (b + Nw)(\alpha(w) - k\t\Phi(w))) = \tilde{A}\Phi(w).
\end{equation}
Different from the approach in  \cite{Ka00}, the feedback $\alpha(w)$ adds additional degrees of freedom to simplify the solution of \eqref{spde21}. This and the initial condition $\Phi(0) = 0$ lead to a Cauchy problem, where the left hand side of \eqref{spde21} vanishes for $w = 0$. Then, the condition of the Cauchy-Kovalevskaya theorem for the solvability of \eqref{spde21} are not satisfied so that \eqref{spde21} becomes a \emph{singular PDE} (see \cite{Ka00}). Therefore, the next lemma clarifies the existence of a unique solution of \eqref{spde21} using \emph{Lyapunov's auxiliary theorem} (for details see \cite{Ka00}).
\begin{thm}[Feedback linearization]\label{lem:singpde1}
Assume that $(\Lambda,b)$ is controllable and that $\tilde{A}$ is a Hurwitz matrix, i.e., its eigenvalues $\tilde{\lambda}_i$, $i =1,\ldots,n$, are in the \emph{Poincar\'e domain}. The Cauchy problem \eqref{spde21} with the IC $\Phi(0) = 0$ has a unique analytic solution $\Phi(w)$ in a neighbourhood of $w = 0$ if the eigenvalues $\lambda_i$, $i = 1,\ldots,n$, of $\mathcal{A}$ and the eigenvalues
$\tilde{\lambda}_i$, $i = 1,\ldots,n$, of $\tilde{A}$ satisfy the \emph{non-resonance condition} $\lambda_i \neq p_1\tilde{\lambda}_1 + \ldots + p_n\tilde{\lambda}_n$ for all $i = 1,\ldots,n$, and for all nonnegative integers $p_i$ such that $p_1 + \ldots + p_n \geq 2$. Moreover, if $(k\t,\tilde{A})$ is observable, then $\Phi(w)$ is a change of coordinates around $w = 0$.
\end{thm}
The proof directly follows from \cite{Ka00}. 
\begin{rem}
A method for the assignment of non-resonant closed-loop eigenvalues $\tilde{\lambda}_i$, $i = 1,\ldots,n$, can be found in \cite{Dev01}. Since the PDE \eqref{spde21} is singular, it cannot be solved with the method of characteristics. However, a power series solution is possible, since Theorem \ref{lem:singpde1} ensures an analytic solution in a neighborhood of the origin. A numerical algorithm, which requires to only solve linear algebraic equations can be found in \cite{Ka00}. In \cite{Deu08b} this method is further refined by making use of a Galerkin approach, which minimizes the remaining nonlinearities on a prespecified multivariable interval in the state space. 
\hfill $\triangleleft$	 
\end{rem}
Introduce  the nonlinear closed-loop operator $\tilde{\mathcal{F}}h = \rho h'' + ah + f(\cdot,h)$, $D(\tilde{\mathcal{F}}) = \{h \in H^2(0,1)\,|\, h'(0) = 0, h'(1) = \alpha(h)-k\t\Phi(h)\} \subset \mathbb{X}$. Then, the \emph{closed-loop Koopman generator} for \eqref{kloop} reads $\tilde{\mathscr{A}}\theta[x] = \mathcal{L}_{\tilde{\mathcal{F}}}\theta[x]$, $\theta[x] \in D(\tilde{\mathscr{A}})$
(cf. \eqref{fablcalc}). Hence, the corresponding \emph{closed-loop Koopman eigenvalue problem} is given by
\begin{equation}\label{clewp}
 \tilde{\mathscr{A}}\tilde{\varphi}_i[x] = \mathcal{L}_{\tilde{\mathcal{F}}}\tilde{\varphi}_i[x] =
 \tilde{\lambda}_i\tilde{\varphi}_i[x], 
\end{equation}
in which $\tilde{\varphi}_i[x] \in D(\tilde{\mathscr{A}})$, $i \in \mathbb{N}$, are the \emph{closed-loop eigenfunctionals} of $\tilde{\mathscr{A}}$ w.r.t. the \emph{closed-loop eigenvalues} $\tilde{\lambda}_i$ (cf. \eqref{kewpp}). The next theorem shows that if the errors $e_x[x]$ and $e_u[x]$ in \eqref{bilin} vanish and $\tilde{A} = \tilde{\Lambda}_n = \diag{\tilde{\lambda}_1,\ldots,\tilde{\lambda}_n}$, then $\bm{\tilde{\varphi}_n}[x] = \col{\tilde{\varphi}_1[x],\ldots,\tilde{\varphi}_n[x]} = \Phi(\bm{\varphi_n}[x])$. This means that the closed-loop eigenfunctionals define a (partially) linearizing change of coordinates $\tilde{w} = \bm{\tilde{\varphi}_n}[x]$ into \eqref{cl} with $e_{\tilde{w}} = 0$. Furthermore, the state feedback \eqref{sf} assigns the eigenfunctionals $\bm{\tilde{\varphi}_n}[x]$ to the closed-loop system similar to what has been shown in \cite{Deu24} for linear parabolic systems.
In fact, the result $\bm{\hat{\tilde{\varphi}}_n}[x] = \Phi(\bm{\varphi_n}[x])$ can be seen as an approximation of the solution for the closed-loop eigenvalue problem \eqref{clewp} giving a system theoretic meaning to the considered feedback linearization approach. For this, the next theorem clarifies that for vanishing errors in \eqref{bilin} the closed-loop eigenfunctionals $\bm{\tilde{\varphi}_n}[x]$ are obtained. 
\begin{thm}[Closed-loop eigenfunctionals]\label{lem:singpde11}\hfill
Assume that $e_x[x] =\linebreak e_u[x] = 0$ in \eqref{bilin}. Then, $\bm{\tilde{\varphi}_n}[x] = \Phi(\bm{\varphi_n}[x])$ is a solution of the closed-loop Koopman eigenvalue problem \eqref{clewp} w.r.t. the eigenvalues of $\tilde{A} = \tilde{\Lambda}_n = \diag{\tilde{\lambda}_1,\ldots,\tilde{\lambda}_n}$ in \eqref{spde21}. 
\end{thm}
\begin{proof}
Differentiating $\Phi(\bm{\varphi_n}[x])$ w.r.t. time yields with \eqref{nsys}, \eqref{sf}, \eqref{spde21} and vanishing errors the result $\d_t\Phi(\bm{\varphi_n}[x]) = \partial_{\bm{\varphi_n}}\Phi(\bm{\varphi_n}[x])\d_t\bm{\varphi_n}[x] = \partial_{\bm{\varphi_n}}\Phi(\bm{\varphi_n}[x])                    (\Lambda_n\bm{\varphi_n}[x] + \rho(b + \linebreak  N\bm{\varphi_n}[x])(\alpha(\bm{\varphi_n}[x])-k\t\Phi(\bm{\varphi_n}[x]))) = \tilde{\Lambda}_n\Phi(\bm{\varphi_n}[x])$. In view of $\Lambda_n\bm{\varphi_n}[x] = \tint_0^1(\delta_x\bm{\varphi_n}[x])(\zeta)\mathcal{F}x(\zeta)\d\zeta$ implied by \eqref{kewpp} the result
$\tint_0^1(\partial_{\bm{\varphi_n}}\Phi(\bm{\varphi_n}[x])\delta_x\bm{\varphi_n}[x])(\zeta)\mathcal{F}x(\zeta)\d\zeta
  + \rho(\partial_{\bm{\varphi_n}}\Phi(\bm{\varphi_n}[x])\delta_x\bm{\varphi_n}[x(t)])(1)(\alpha(\bm{\varphi_n}[x])-k\t\Phi(\bm{\varphi_n}[x]))
  = \tilde{\Lambda}_n\Phi(\bm{\varphi_n}[x])$ also holds. Then, $\partial_{\bm{\varphi_n}}\Phi(\bm{\varphi_n}[x])\delta_x\bm{\varphi_n}[x] = \delta_x\Phi(\bm{\varphi_n}[x])$ yields $\tint_0^1\delta_x\Phi(\bm{\varphi_n}[x](\zeta)\tilde{\mathcal{F}}x(\zeta)\d\zeta = \tilde{\Lambda}_n\Phi(\bm{\varphi_n}[x])$ so that $\bm{\tilde{\varphi}_n}[x] = \Phi(\bm{\varphi_n}[x])$ follows verifying \eqref{clewp}.
\end{proof}	



The next lemma shows that the overall transformation $\Theta : \mathbb{C}^n \times \ell^2 \to \mathbb{R}^n$ given by
$\tilde{w} = \Phi(\bm{\varphi_n}(\bm{x})) = \Theta(\bm{x_n},\bm{x_\infty})$ with  $\bm{x_n} = (x_{1},\ldots,x_{n},) \in \mathbb{C}^n$ and $\bm{x_\infty} = (x_{n+1},x_{n+2},\ldots) \in \ell^2$, $x_i = \langle x,\phi_i\rangle$, is a local change of coordinates (cf. \eqref{wtiltraf}). Then, there exists a nonlinear transformation $\Theta^{-1} : \mathbb{C}^n \times \ell^2 \to \mathbb{C}^n$ such that
\begin{equation}\label{xexpand}
 x = \bm{\phi_n}\t\Theta^{-1}(\tilde{w},\bm{x_\infty}) + \sum_{i=n+1}^\infty\langle x,\phi_i\rangle\phi_i
\end{equation}
with $\bm{\phi_n}(z) = \col{\phi_1(z),\ldots,\phi_n(z)} \in \mathbb{C}^n$ verifying that \eqref{modtraf} is also a change of coordinates.
\begin{Lemma}[Change of coordinates]\label{lem:coord}\hfill
There exists the inverse transformation 
\begin{equation}\label{invtraf}
 \bm{x_n} = \Theta^{-1}(\tilde{w},\bm{x_{\infty}})
\end{equation}
if $\Phi(w)$, $\Phi(0) = 0$, is a local change of coordinates around the origin and $\bm{\varphi_n}(\bm{x})$, $\bm{\varphi_n}(\bm{0}) = \bm{0}$, satisfies $\det\partial_{\bm{x_n}}\bm{\varphi_n}(\bm{0}) \neq 0$. If $\Theta(\bm{x})$ is a $C^m$-map (i.e., $m$-times continuously Fr$\acute{\text{e}}$chet differentiable) in a neighborhood of $(\bm{0_n},\bm{0})$, then so is $\Theta^{-1}(\tilde{w},\bm{x_{\infty}})$.
\end{Lemma}	
\begin{proof}
Consider the mapping $F(\tilde{w},\bm{x_n},\bm{x_\infty}) \!=\! \tilde{w} \! -\! \Theta(\bm{x_n},\bm{x_{\infty}}) \!:\! \mathbb{C}^n \!\times\! \mathbb{C}^n \!\times\! \ell^2 \to \mathbb{C}^n$ with	$F(\bm{0},\bm{0},\bm{0}) = \bm{0}$. The partial Fr$\acute{\text{e}}$chet derivative $F_{\bm{x}_n}$ of $F$ is $F_{\bm{x_n}}(\tilde{w},\bm{x_n},\bm{x_\infty}) = \Theta_{\bm{x_n}}(\bm{x_n},\bm{x_\infty}) = \partial_w\Phi(w)\partial_{\bm{x_n}}\bm{\varphi_n}(\bm{x})$ in view of $\tilde{w} = \Theta(\bm{x_n},\bm{x_\infty})$ (see \cite{Ven21}) and $\det F_{\bm{x_n}}(\bm{0},\bm{0},\bm{0}) =  \det\partial_w\Phi(\bm{0})\partial_{\bm{x_n}}\bm{\varphi_n}(\bm{0}) \neq 0$, because $\Phi(w)$ is a local change of coordinates around $w = 0$ and $\det\partial_{\bm{x_n}}\bm{\varphi_n}(\bm{0}) \neq 0$, i.e., $F_{\bm{x_n}}(\bm{0},\bm{0},\bm{0})$ is boundedly invertible. Since $F$ and $F_{\bm{x_n}}$ are continuous at $(\bm{0},\bm{0},\bm{0})$, Theorem 4.B in \cite{Ze86} shows that \eqref{invtraf} exists. Furthermore, if $F$ is a $C^m$-map on a neighborhood of $(\bm{0},\bm{0},\bm{0})$, then \eqref{invtraf} is also a $C^m$-map on a neighborhood of $(\bm{0},\bm{0},\bm{0})$.
\end{proof}	
The following theorem presents the conditions to ensure a stable origin for the closed-loop system \eqref{cl}.
\begin{thm}[Closed-loop stability]\label{thm:clstab}\hfill
Introduce the state space $\mathbb{X}_{\text{cl}} = \mathbb{C}^n\oplus\ell_2$ for \eqref{kloop} endowed with the norm $\|\cdot\|_{\text{cl}} = (\|\cdot\|^2_{\mathbb{C}^n} + \|\cdot\|^2_{\ell_2})^{1/2}$. Assume that for the error $e_{\tilde{w}}(\tilde{\omega},x)$ in \eqref{etilde} the uniform Lipschitz condition
\begin{equation}\label{Lip1}
\|e_{\tilde{w}}(h_1,x) - e_{\tilde{w}}(h_2,x)\|_{\mathbb{C}^n} \leq L_1\|h_1 - h_2\|_{\mathbb{C}^n}
\end{equation}
for all  $x \in \{h \in D(\theta) = D(\mathcal{F}) \,|\, \|h\| \leq r\}$,  $\|h_i\|_{\mathbb{C}^n} \leq r$, $i = 1,2$, with a nonnegative constant $L_1(r)$ holds and that the inverse transformation $\Theta^{-1}(\tilde{w},\bm{x_{\infty}})$ in \eqref{invtraf} satisfies the Lipschitz condition
\begin{equation}\label{thetalip}
\|\Theta^{-1}(h_1) - \Theta^{-1}(h_2)\|_{\mathbb{C}^n} \leq L_2 \|h_1 - h_2\|_{\text{cl}}
\end{equation}
with $L_2(r) > 0$ for all  $h_i \in B_r = \{h \in \mathbb{X}_{\text{cl}} \,|\, \|h\|_{\text{cl}} \leq r\}$, $i = 1,2$. In addition, $\bm{\varphi_n}[x]$ is required to be locally Lipschitz around the origin.

If $\alpha_0 = \max_{\lambda \in \sigma(\tilde{A})}\operatorname{Re}\lambda+ L_1^2 < 0$ with  $\epsilon > 0$  be such that $\alpha_{\text{cl}} = \alpha_0  + \epsilon$ satisfies $\lambda_{n+1} < \alpha_{\text{cl}} < 0$, then the origin of the closed-loop system \eqref{kloop} is locally exponentially stable for $\|x(0)\| \leq cr$, $0 < c < 1$, i.e., $\|x(t)\| \leq M\e^{\tilde{\alpha}_{\text{cl}}t}\|x(0)\|$, $t \geq 0$, where $M \geq 1$ and $0 < \tilde{\alpha}_{\text{cl}} < \alpha_{\text{cl}}$.	
\end{thm}	
\begin{proof}
The solution of \eqref{finlin} satisfies $\|\tilde{w}(t)\|^2_{\mathbb{C}^n} \leq 2(c_1\e^{2\alpha t}\|\tilde{w}(0)\|^2_{\mathbb{C}^n} + L_1^2\int_0^t\e^{2\alpha(t-\tau)}\|\tilde{w}(\tau)\|^2_{\mathbb{C}^n}\d\tau)$, in which \eqref{Lip1}, $\alpha = \alpha_{\text{cl}} - L_1^2$ and $c_1 > 0$ have been utilized. By making use of the Gronwall inequality (see, e.g., \cite[App. A]{Ka02}), it is readily verified that
\begin{equation}\label{wtilsol}
	\|\tilde{w}(t)\|^2_{\mathbb{C}^n} \leq 2c_1\e^{2\alpha_{\text{cl}}t}\|\tilde{w}(0)\|^2_{\mathbb{C}^n}, \quad t \geq 0,
\end{equation}
holds. The solution of \eqref{infsub} can be estimated by
\begin{align}
	\|\bm{x_\infty}(t)\|_{\ell_2}^2 &\leq  2\e^{2\lambda_{n+1}t}\|\bm{x_\infty}(0)\|_{\ell_2}^2\nonumber\\
	&\quad + 4c_2C^2L_2^2\int_0^t\e^{2\lambda_{n+1}(t-\tau)}\|x_{\text{cl}}(\tau)\|^2_{\text{cl}}\d\tau\nonumber\\
	&\quad + 4c_3\sum_{i=n+1}^{\infty}\int_0^t\e^{2\lambda_i(t-\tau)}\|\tilde{w}(\tau)\|^2_{\mathbb{C}^n}\d\tau\label{xinfsol} 
\end{align}
in view of \eqref{Lcond} and \eqref{thetalip} for $c_2, c_3 > 0$. Using \eqref{wtilsol} in the last line of \eqref{xinfsol} gives
$8c_1c_3\sum_{i=n+1}^{\infty}\int_0^t\e^{2\lambda_i(t-\tau)}\|\tilde{w}(\tau)\|^2_{\mathbb{C}^n}\d\tau
\leq 4c_1c_3\sum_{i=n+1}^{\infty}(\e^{2\alpha_{\text{cl}}t}/(|\alpha_{\text{cl}} - \lambda_i|)\|\tilde{w}(0)\|^2_{\mathbb{C}^n}
\leq c_4\e^{2\alpha_{\text{cl}}t}\|\tilde{w}(0)\|^2_{\mathbb{C}^n}$ for finite $c_4 > 0$, because $|\lambda_i| \in O(i^2)$ holds for the Sturm-Liouville operator $-\mathcal{A}$ (see \cite{Orl17a}) so that the infinite sum is convergent. By combining this result with \eqref{wtilsol}, \eqref{xinfsol} and taking $\lambda_{n+1} < \alpha_{\text{cl}}$ into account, a simple calculation yields
\begin{multline}\label{xcdsol}
	\|x_{\text{cl}}(t)\|^2_{\text{cl}} \leq M_{\text{cl}}(\e^{2\alpha_{\text{cl}}t}\|x_{\text{cl}}(0)\|_{\text{cl}}^2\\ + C^2\int_0^t\e^{2\alpha_{\text{cl}}(t-\tau)}\|x_{\text{cl}}(\tau)\|_{\text{cl}}^2\d\tau)
\end{multline}
in $B_r$ for $t \geq 0$, $M_{\text{cl}} \geq 1$ and $x_{\text{cl}} = \col{\tilde{w},x_\infty}$. Assume that $x_{\text{cl}}(t) \in B_r$, $0 \leq t < t_1$. Then, the escape time $t_1$ from the set $B_r$ can either be finite or infinite. In order to prove that only the last case is possible for $r$ sufficiently small consider $r^2 = \|x_{\text{cl}}(t_1)\|^2_{\text{cl}}$, which leads  to $r^2 = M_{\text{cl}}\e^{2\alpha_{\text{cl}}t_1}(\|x_{\text{cl}}(0)\|_{\text{cl}}^2 \!+\! C^2\int_0^{t_1}\e^{-2\alpha_{\text{cl}}\tau}\|x_{\text{cl}}(\tau)\|_{\text{cl}}^2\d\tau)$ in view of \eqref{xcdsol}. With $\|x_{\text{cl}}(t)\|_{\text{cl}}^2 \leq r^2$, $0 \leq t < t_1$, the result
\begin{align}\label{p2:est}
	r^2 &\leq M_{\text{cl}}\e^{2\alpha_{\text{cl}} t_1}(\|x_{\text{cl}}(0)\|^2_{\text{cl}} + \textstyle\frac{C^2r^2}{2\alpha_{\text{cl}}}
	(1 - \e^{-2\alpha_{\text{cl}} t_1}))\nonumber\\
	&\leq M_{\text{cl}}(\|x_{\text{cl}}(0)\|^2_{\text{cl}} + \textstyle\frac{C^2r^2}{2\alpha_{\text{cl}}})
\end{align}
follows from \eqref{xcdsol} for $t = t_1$. If $\|x_{\text{cl}}(0)\|^2_{\text{cl}} \leq (r^2q)/M_{\text{cl}}$, $0 < q < 1$, then $x_{\text{cl}}(0) \in B_r$ in view of  $M_{\text{cl}} \geq 1$. With this, \eqref{p2:est} becomes $r^2 \leq (q + (C^2M_{\text{cl}}/(2\alpha_{\text{cl}}))r^2$. If $C^2 < ((1-q)2\alpha_{\text{cl}})/M_{\text{cl}}$, that is possible for sufficiently small $r$ in view of \eqref{Lcond}, then $r^2 < (q + 1-q)r^2 = r^2$ which is a contradiction. Thus, the escape time $t_1$ is infinite. This means that for sufficiently small $r$ the state $x_{\text{cl}}(t)$ remains in $B_r$ for $t \geq 0$ if $\|x_{\text{cl}}(0)\|_{\text{cl}} \leq r\sqrt{q/M_{\text{cl}}}$. Then, the local Lipschitz conditions \eqref{Lcond}, \eqref{Lip1} and \eqref{thetalip} are valid for $t \geq 0$. Under this condition it is shown next that the state $x_{\text{cl}}$ converges exponentially to the origin. Since $x_{\text{cl}}(0) \in B_r$ implies $x_{\text{cl}}(t) \in B_r$, $t \geq 0$, one can use \eqref{xcdsol} for $t \geq 0$ so that with $C^2M_{\text{cl}} < 2(1-q)\alpha_{\text{cl}}$ the result $\e^{-2\alpha_{\text{cl}} t}\|x_{\text{cl}}(t)\|_{\text{cl}}^2 \leq M_{\text{cl}}(\|x_{\text{cl}}(0)\|_{\text{cl}}^2 + \int_0^{t}2(1-q)\alpha_{\text{cl}}\e^{-2\alpha_{\text{cl}}\tau}\|x_{\text{cl}}(\tau)\|_{\text{cl}}^2\d\tau)$ follows. By using the \emph{Gronwall's Lemma} (see, e.g., \cite[Lem. A.5.30]{Cu20}) it can be verified that $\e^{-2\alpha_{\text{cl}} t}\|x_{\text{cl}}(t)\|_{\text{cl}}^2 \leq M_{\text{cl}}\|x_{\text{cl}}(0)\|^2_{\text{cl}}\exp(\int_0^t2(q-1)\alpha_{\text{cl}}\d\tau)\linebreak = M_{\text{cl}}\|x_{\text{cl}}(0)\|_{\text{cl}}^2\e^{2(q-1)\alpha_{\text{cl}}t}$ holds, which yields $\|x_{\text{cl}}(t)\|_{\text{cl}}^2 \leq  M_{\text{cl}}\e^{2q\alpha_{\text{cl}}t}\|x_{\text{cl}}(0)\|_{\text{cl}}^2$. Since $0 < q < 1$, this shows $\|x_{\text{cl}}(t)\|_{\text{cl}} \leq \sqrt{M_{\text{cl}}}\e^{\tilde{\alpha}_{\text{cl}}t}\|x_{\text{cl}}(0)\|_{\text{cl}}$, $t \geq 0$, where $\sqrt{M_{\text{cl}}} \geq 1$ and $0 < \tilde{\alpha}_{\text{cl}} < \alpha_{\text{cl}}$. Since the sequence $\{\phi_i, i \in \mathbb{N}\}$ is an orthonormal Riesz basis for $\mathbb{X}$ (see \cite{Del03}), the \emph{Parseval's equality} $\|\sum_{i=1}^{\infty}h_i\phi_i\|^2 = \sum_{i=1}^{\infty}|h_i|^2 = \|h\|^2_{\ell_2}$ is satisfied. This, $M_{\text{cl}} \geq 1$ and $0 < q < 1$ lead to $\|x(0)\| \leq cr$, $0 < c <1$, implied by $\|x_{\text{cl}}(0)\|_{\text{cl}} \leq \sqrt{q/M_{\text{cl}}}r$. Then,  \eqref{xexpand}, \eqref{thetalip} as well as the fact that both the transformation $\Phi$ and $\bm{\varphi_n}[x]$ are locally Lipschitz around the origin can be used to show exponential stability by means of a straightforward calculation.
\end{proof}	

\section{Data-Driven Solution of the Koopman Eigenvalue Problem}\label{sec:dewp}
\subsubsection{Extended Dynamic Mode Decomposition} One of the advantages when using the eigenfunctionals for the data-driven control is that they can be determined on the basis of the Koopman operator $\mathscr{K}(t)$ rather than the Koopman generator $\mathscr{A}$ (cf. \eqref{Kopewp} and \eqref{kewpp}). This avoids to use state derivative data for the data-driven solution of the corresponding eigenvalue problem as it is necessary for determining the Koopman generator $\mathscr{A}$ (see, e.g., \cite{Klus20} for the ODEs). Therefore, in what follows the \emph{extended dynamic mode decomposition (eDMD)} for determining the K-eigenvalues $\lambda_i$, $i \in \mathbb{N}$, w.r.t. the eigenfunctionals $\varphi_i[x]$ on the basis of the approximation of the Koopman operator $\mathscr{K}(t)$ in \eqref{kewpp} is shortly reviewed. The presented results follow the approach already given in \cite{Mau21}, which is a generalization of the results in \cite{Wi15} for finite-dimensional systems. More details concerning DMD and eDMD can be found in \cite{Ku16}.

For the data-driven approach to determine the eigenvalues and \linebreak -functionals of $\mathscr{K}(t)$  and thus of $\mathscr{A}$ a set of linearly independent basis functionals $\psi_i[x] \in D(\mathscr{A})$, $i = 1,\ldots,L$, is introduced and collected in the \emph{dictionary} $\mathcal{D} = \{\psi_1[x],\ldots,\psi_L[x]\}$ so that $\mathbb{O}_L = \operatorname{span}\mathcal{D} \subset \mathbb{O}$. Then, any observable $\theta[x] \in \mathbb{O}_L$ can be uniquely represented by $\theta[x] = \sum_{i=1}^La_i\psi_i[x] = \kappa\t\col{\psi_1[x],\ldots,\psi_L[x]} =  \kappa\t\psi[x]$ with $\kappa \in \mathbb{R}^L$ and the \emph{feature map} $\psi[x]: D(\theta) = D(\mathcal{F}) \to \mathbb{R}^L$, that maps elements of $D(\theta)$ to the \emph{feature space} $\mathbb{R}^L$. With this, the value of $\theta[x] \in \mathbb{O}_L$ at a \emph{sampling time} $t_s > 0$ is 
\begin{equation}\label{Kadv}
 \mathscr{K}(t_s)\theta[x] \!=\! \theta[\mathcal{T}_{\mathcal{F}}(t_s)x] \!=\! \kappa\t\psi[x(t_s)] \!=\! \kappa\t\hat{\mathscr{K}}\psi[x] \!+\! r(x),
\end{equation}
in which $\hat{\mathscr{K}} \in \mathbb{C}^{L \times L}$ is the \emph{Koopman matrix}. This is a finite-dimensional approximation for the Koopman operator $\mathscr{K}(t_s)$ at $t = t_s$. Since, in general, $\mathbb{O}_L$ is not an invariant subspace for $\mathscr{K}(t)$, there appears a \emph{residual} $r(x)$. Hence, $\hat{\mathscr{K}}$ is only an approximation for $\mathscr{K}(t)$, whose accuracy is quantified by the residual $r(x)$. Consequently, $\hat{\mathscr{K}}$ has to be determined to render $r$ sufficiently small. Since the semiflow $\mathcal{T}_{\mathcal{F}}(t)$ is not known, this is achieved by making use of available data. For this, consider the \emph{snapshot pairs} $\{(x^i,\mathcal{T}_{\mathcal{F}}(t_s)x^i)\}_{i = 1}^M$, in which $x^i$ is a state trajectory at time instant $t_i$. Then, the data pairs $\{(\psi_j[x^i],\psi_j[\mathcal{T}_{\mathcal{F}}(t_s)x^i])\}_{i = 1}^M$, $j = 1,\ldots,L$, can be introduced, which are obtained from simulations or from measurement data. By defining $(x^i)^+ = \mathcal{T}_{\mathcal{F}}(t_s)x^i$ and $\underline{x} = \col{x^1,\ldots,x^M}$ the relation \eqref{Kadv} evaluated for the considered data takes the form $\kappa\t\Psi(\underline{x}^+) = \kappa\t\hat{\mathscr{K}}\Psi(\underline{x}) + \underline{r}\t(\underline{x})$ with the data matrix
$\Psi(\underline{x}) = [\psi[x^1], \ldots, \psi[x^M] \in \mathbb{R}^{L \times M}$ and the residual vector $\underline{r}(\underline{x}) = \col{r(x^1),\ldots,r(x^M)} \in \mathbb{R}^M$. For solving this linear regression problem, the error $\underline{r}(\underline{x})$ is minimized by the least-square solution
\begin{equation}\label{matK3}
 \hat{\mathscr{K}} = \Psi(\underline{x}^+)\Psi\t(\underline{x})(\Psi(\underline{x})\Psi\t(\underline{x}))^{\dagger},
\end{equation}
where $(\cdot)^{\dagger}$ is the Moore-Penrose inverse. 
\begin{rem}
For finite-dimensional nonlinear systems convergence of the eDMD in the large data limit is verified in \cite{Kor18}, while a probabilistic finite-data error bound is derived in \cite{Nue23,Sch23}. It is expected that under suitable conditions, also the generalized eDMD shows similar convergence properties. However, this has not been investigated in the literature so far and is out of the scope for this contribution. 
	\hfill $\triangleleft$	 
\end{rem}

\subsubsection{Approximation of the K-Eigenvalues and -functionals} Once the Koopman matrix $\hat{\mathscr{K}}$ has been determined, an approximation of the K-eigenvalues and -functionals can be easily obtained. Consider
\begin{equation}\label{Kopewpapr}
\mathscr{K}(t_s)\theta_i[x] = \mathscr{K}(t_s)w_i\t\psi[x] = w_i\t\hat{\mathscr{K}}\psi[x] + r_i(x)
\end{equation} 
for $\theta_i[x] \in \mathbb{O}_L$, where $w_i\t$ follows from solving the eigenvalue problem $w\t_i\hat{\mathscr{K}} = \mu_iw\t_i$, $i = 1, \ldots, L$. Inserting this in \eqref{Kopewpapr} yields $\mathscr{K}(t_s)w_i\t\psi[x] = \mu_iw_i\t\psi[x] + r_i(x)$, $i = 1,\ldots,L$. Consequently, an approximation of the eigenfunctionals is given by
\begin{equation}\label{approxefktl}
\hat{\varphi}_i[x] =  w_i\t\psi[x], \quad i = 1,\ldots,L,
\end{equation}
w.r.t. the approximate K-eigenvalue $\mu_i$. In view of \eqref{Kopewp}, the corresponding approximation for the eigenvalues of the Koopman generator $\mathscr{A}$ is $\hat{\lambda}_i = \ln \mu_i/t_s$.

\subsubsection{Approximation of $\rho$, b and N}
With the state-derivative sample $\dot{x}(z,t_0)$ for $u(t) = u_0$, $t \geq 0$, one can determine $\rho$ from \eqref{nsys} by $\hat{\rho} = (\d_t\hat{\varphi}_i[x^j] - \hat{\lambda}_i\hat{\varphi}_i[x^j])(\delta_x\hat{\varphi}_i[x^j](1)u_0)^{-1}$ for some open-loop eigenfunctional $\hat{\varphi}_i[x]$, $i = 1,\ldots,L$, at some snapshot $x^j$, $j = 1,\ldots,M$. The time derivative $\dot{x}(z,t_0)$  can be obtained from the state data using, e.g., finite differences. In the presence of noise the method proposed in \cite{Br16a} can be applied.

As an alternative to the analytic approach in Theorem \ref{lem:optbilin} for optimal bilinearization, the following data-driven approach to determine $b$ and $N$ using the obtained open-loop eigenfunctionals $\hat{\varphi}_i[x]$ can be used. For this, consider the error $e_u[x] =  \rho(\delta_x\bm{\varphi_n}[x])(1) - [b,N]\bm{\bar{\varphi}_n}[x]$, where $\bm{\bar{\varphi}_n}[x] = \col{1,\bm{\varphi_n}[x]}$.
Evaluating this error at the data yields
\begin{multline}
	E_u[\underline{x}] = 
	\underbrace{\begin{bmatrix}
			\hat{\rho}\delta_x\bm{\hat{\varphi}_n}[x^1](1) & \ldots & \hat{\rho}\delta_x\bm{\hat{\varphi}_n}[x^M](1)
	\end{bmatrix}}_{\mathbb{D}[\underline{x}]}\\
	-  
	\begin{bmatrix}
		b & N
	\end{bmatrix}
	\underbrace{\begin{bsmallmatrix}
			1                  & \ldots & 1\\
			\bm{\hat{\varphi}_n}[x^1]  & \ldots & \bm{\hat{\varphi}_n}[x^M]
	\end{bsmallmatrix}}_{\mathbb{F}[\underline{x}]},
\end{multline}
in which the approximations $\hat{\varphi}_i[x]$  and $\hat{\rho}$ have been utilized. Hence, the solution minimizing the Frobenius norm $\|E_u[\underline{x}]\|_F$ is given by $[b,N] =  \mathbb{D}[\underline{x}]\mathbb{F}^\dagger[\underline{x}]$. Note that the error $E_u[\underline{x}]$ can be evaluated using only data. Hence, a small $\|E_u[\underline{x}]\|_F$ is an indicator for a suitable choice of $M$. The resulting error $e_u[x]$ and also $e_x[x]$ following from the approximation of the Koopman eigenvalues and eigenfunctionals is taken into account in the stability analysis of Theorem \ref{thm:clstab}. Then, an eigenvalue assignment for the matrix $\tilde{A}$ ensuring $\max_{\lambda \in \sigma(\tilde{A})}\operatorname{Re}\lambda+ L_1^2 < 0$ yields a closed-loop system with a locally exponentially stable origin.

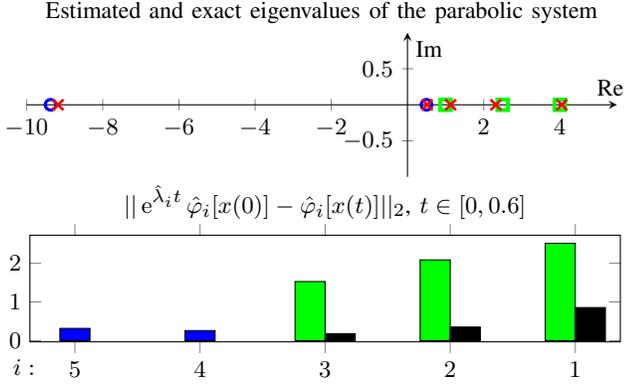
\begin{figure}[t]
	\begin{tikzpicture}
	\begin{axis}[yticklabel style={text width=3em,align=right},
	title={Estimated and exact eigenvalues of the parabolic system},
	every axis title shift = 0.1mm,
	height=3.5cm,
	width = 0.52\textwidth, enlarge x limits=0.05, ybar=5mm,
	xmax=5.5,
	xmin=-10,
	axis lines = middle,
	every axis x label/.append style={xshift=2mm, yshift=0mm},
	xlabel=Re,
	ylabel=Im,
	ytick={-0.5, 0.5}
	]
	\addplot [only marks,mark=o, blue, line width=1.2pt] table{data/realEvs.dat};\label{realEv}
	\addplot [only marks,mark=square, green, opacity=1, line width=1.5pt] coordinates{(1, 0) (2.5, 0) (4, 0)};
	0		 						\addplot [only marks,mark=x,mark size=1mm, red, opacity=0.8, line width=1.0pt] table{data/dmdEvs.dat};\label{dmdEv}
	\end{axis};	
	\end{tikzpicture}\\
	\vspace{-4mm}\\\hspace*{-8.4mm}
	\begin{tikzpicture}
	\begin{axis}[yticklabel style={text width=3em,align=right},
	ybar,
	title={$||\operatorname{e}^{\hat{\lambda}_i t}\hat{\varphi}_i[x(0)] - \hat{\varphi}_i[x(t)]||_2$, $t \in [0,0.6]$},
	title style={yshift=1mm},
	every axis title shift = 0.1mm,
	height=3cm,
	width = 0.52\textwidth,
	xlabel=$i:$,
	every axis x label/.append style={xshift=-40mm, yshift=4.5mm},
	enlargelimits=0.09,
	xtick={5, 4, 3, 2, 1},
	x dir=reverse,]
	\addplot[bar shift=-2mm, fill=green, bar width=4mm]  table[x=t, y = y1]{data/phiErr.dat};
	\addplot[bar shift=2mm, fill=black, bar width=4mm]  table[x=t, y = y1]{data/phiErr3.dat};
	\addplot[bar shift=0mm, fill=blue, bar width=4mm]  table[x=t, y = y1]{data/phiErr2.dat};
	\end{axis};	
	\end{tikzpicture}
	\vspace{-2mm}
	\caption{The top plot shows the eigenvalues $\hat{\lambda}_i$, $i = 1,\ldots,5$ (\hspace{-1mm}\redcross\hspace{-1mm}) resulting from the eDMD as well as the exact principal eigenvalues $\lambda_i$, $i = 4,5$ \linebreak (\hspace{-1.0mm}\bluecircle\hspace{-1.0mm}) and non-principal eigenvalues $\lambda_i = \mu_i\lambda_4$, $\mu_i \in \mathbb{N}$, $i = 1,2,3$ \linebreak (\hspace{-1mm}\greencube\hspace{-1mm}) . The bottom plot displays the $L_2$-norms of the prediction error of the eigenpairs $(\hat{\lambda}_i, \hat{\varphi}_i[x]), i = 1, \dots, 5$. For the non-principal eigenvalues also $||\operatorname{e}^{\hat{\lambda}_i t}\hat{\varphi}^{\mu_i}_4[x(0)] - \hat{\varphi}^{\mu_i}_4[x(t)]||_2, i = 1, 2, 3$, is shown in black, since $\varphi_i[x] = \varphi^{\mu_i}_4[x]$.\label{fig:eDMD}}
\end{figure}

\section{Example}\label{sec:ex}
The results of the paper are illustrated for the stabilization of the unstable and thus hyperbolic equilibrium $x = 0$ of a semilinear parabolic system
\eqref{plant} with $\rho = 1$, $a(z) = 0.5$ and $f(z,x) = 0.5x^2$, that exhibits a finite-time blow up.
The system is assumed to be unknown, but state data is available for the controller design. 

\subsubsection{Extended Dynamic Mode Decomposition}
For the data-driven system analysis $L = 27$ basis functionals $\psi_{ikl}[x] = \langle f_i,x^k\rangle^l$, $i,k,l = 1,2,3$, with $f_1(z) = 1$, $f_i(z) = \sqrt{2}\cos(i\pi z)$ are considered, where the powers $l$ account for non-principal eigenvalues (see Remark \ref{rem:prineig}). The data pairs $\{\psi_{ikl}[x_0],\psi_{ikl}[\mathcal{T}_{\mathcal{F}}(t_s)x_0]\}$ with $t_s = 0.1$ were generated in simulations. For this, the function $g(z) = -(a/(\pi\omega)^3)(((\omega\pi)^2(z-1)z-2)\sin(\omega\pi z + \pi \phi_0) + \omega\pi(2z-1)\cos(\omega\pi z + \pi\phi_0)$ with $a = 0.8$, $\omega = 0.42$ and $\phi_0 = \frac{2}{3}$ is introduced, that satisfies the Neumann BCs of \eqref{plant}. Then, with $g_{i} = \langle g,h_i\rangle$ and $h_1(z) = 1$, $h_i(z) = \sqrt{2}\cos(i\pi z)$, $i = 2,\ldots,5$, the ICs $x_0 = \sum_{i=1}^5x_{0,i}h_i(z)$ are generated, where $x_{0,i}$ is uniformly distributed in the set $[g_i-0.04,g_i+0.04]$ to obtain informative data. Then, $M=200$ simulations are used to determine the data matrices $\Psi(\underline{x}) \in \mathbb{R}^{27\times 200}$ and $\Psi(\underline{x}^+) \in \mathbb{R}^{27\times 200}$ so that the Koopman matrix $\hat{\mathscr{K}} \in \mathbb{R}^{27 \times 27}$ follows from \eqref{matK3}. For this the MATLAB function \texttt{pdepe} with 101 spatial sampling points was applied to simulate the semilinear parabolic PDE.  By solving the eigenvalue problem for $\hat{\mathscr{K}}$ in \eqref{Kopewpapr} the estimates $\hat{\lambda}_i$ for the eigenvalues and the corresponding eigenfunctionals \eqref{approxefktl} are computed. The result is shown in Figure \ref{fig:eDMD} (top), where the five largest exact eigenvalues and their estimates are plotted. They consist of the principal eigenvalues $\hat{\lambda}_i$, $i = 4,5$, and the non-principal eigenvalues $\lambda_i = \mu_i\lambda_4$, $i = 1,2,3$, with  $\mu_1 = 8, \mu_2 = 4, \mu_3 = 2$ (cf. Remark \ref{rem:prineig}). The latter appear due to the nonlinear basis functionals. For the data-driven validation of the eigenvalues and -functionals the $L_2$-norm of the prediction errors are plotted in Figure  \ref{fig:eDMD} (bottom) by evaluating the eigenfunctionals on a test trajectory with the IC $x(z, 0) = 0.8$ for the time interval $t\in[0, 0.6]$. The result shows that the principal eigenvalues and -functionals are captured very well and can be used to define the Koopman modal representation \eqref{nsys}. 


\begin{figure}[t]
	\hspace{-7.0mm}\begin{tabular}{l r}
		\begin{tikzpicture}
		\begin{axis}[yticklabel style={text width=3em,align=right},
		title={$\Phi_1(w)({\scalebox{0.8}{\ref{r}}}), \tilde{w}^{\text{lin}}_1({\scalebox{0.8}{\ref{y}}})$},
		every axis title shift = 0.1mm,
		height=3cm,
		width = 0.29\textwidth,
		xmax=5,
		xmin=0,
		ymin=0]
		\addplot [blue, line width=1.2pt] table[x=t, y = y31]{data/testControl.dat};\label{r}
		\addplot [green, densely dashed, line width=1.2pt] table[x=t, y = y11]{data/testControl.dat};\label{y}
		\end{axis};	
		\end{tikzpicture}&\hspace{-11mm}
		\begin{tikzpicture}
		\begin{axis}[yticklabel style={text width=3em,align=right},
		title={$\Phi_1(\boldsymbol{\hat{\varphi}_2}[x])({\scalebox{0.8}{\ref{r}}}), \tilde{w}^{\text{lin}}_1({\scalebox{0.8}{\ref{y}}})$},
		every axis title shift = 0.1mm,
		height=3cm,
		width = 0.29\textwidth,
		ymin=0,
		xmax=5,
		xmin=0]
		\addplot [blue, line width=1.2pt] table[x=t, y = y21]{data/testControl2.dat};
		\addplot [green, densely dashed, line width=1.2pt] table[x=t, y = y11]{data/testControl2.dat};
		\end{axis};	
		\end{tikzpicture}\\
		\begin{tikzpicture}
		\begin{axis}[yticklabel style={text width=3em,align=right},ytick={0.06, 0.04, 0.02, 0}, y tick label style={/pgf/number format/fixed},scaled y ticks = false,
		title={$\Phi_2(w)({\scalebox{0.8}{\ref{r}}}), \tilde{w}^{\text{lin}}_2({\scalebox{0.8}{\ref{y}}})$},
		every axis title shift = 0.1mm,
		height=3cm,
		width = 0.29\textwidth,
		xlabel=$t$,
		every axis x label/.append style={xshift=0mm, yshift=4.0mm},
		xmax=5,
		xmin=0,]
		\addplot [blue, line width=1.2pt] table[x=t, y = y32]{data/testControl.dat};
		\addplot [green, densely dashed, line width=1.2pt] table[x=t, y = y12]{data/testControl.dat};
		\end{axis};	
		\end{tikzpicture}&\hspace{-10mm}
		\begin{tikzpicture}
		\begin{axis}[yticklabel style={text width=3em,align=right},
		title={$\Phi_2(\boldsymbol{\hat{\varphi}_2}[x])({\scalebox{0.8}{\ref{r}}}), \tilde{w}^{\text{lin}}_2({\scalebox{0.8}{\ref{y}}})$},
		every axis title shift = 0.1mm,
		height=3cm,
		width = 0.29\textwidth,
		xlabel=$t$,
		every axis x label/.append style={xshift=0mm, yshift=4.0mm},
		xmax=5,
		xmin=0]
		\addplot [blue, line width=1.2pt] table[x=t, y = y22]{data/testControl2.dat};
		\addplot [green, densely dashed, line width=1.2pt] table[x=t, y = y12]{data/testControl2.dat};
		\end{axis};	
		\end{tikzpicture}
	\end{tabular}
	\vspace{-1mm}
	\caption{The left column shows the comparison of the transformed state trajectory $\Phi(w)$({\scalebox{0.8}{\ref{r}}}) for the closed-loop
		  bilinear system  \eqref{blinappr}, \eqref{sf} and the IC $w(0) = \boldsymbol{\hat{\varphi}_2}[2.4g]$ with the corresponding solution $\tilde{w}^{\text{lin}}$  ({\scalebox{0.8}{\ref{y}}}) of the linear target system \eqref{lintag}. The right column compares the transformed state trajectories $\Phi(\boldsymbol{\hat{\varphi}_2}[x])$ ({\scalebox{0.8}{\ref{r}}}) of the closed-loop PDE system \eqref{plant}, \eqref{sf} for the IC $x_0(z) = 3.1g(z)$ and the corresponding solution of the linear target system \eqref{lintag}.\label{fig:lin}}
\end{figure}
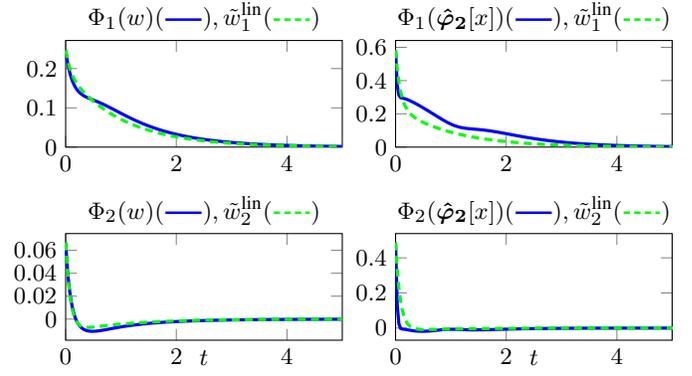
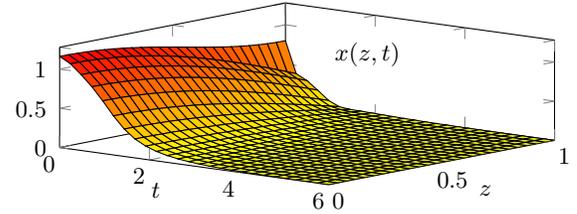
\begin{figure}[h!]
	\centering
	\begin{tikzpicture}	
	\begin{axis}
	[title={$x(z,t)$},
	title style={yshift=-10mm, xshift=8mm},
	width=0.45\textwidth,
	height=40mm,
	view/h = 220,
	every axis x label/.append style={xshift=0mm, yshift=3mm},
	every axis y label/.append style={xshift=0mm, yshift=3mm},
	every axis title/.append style={at={(0.5, 0.93)}},
	x dir=reverse,
	y dir=reverse,
	xlabel={$t$},
	ylabel={$z$},
	every x tick label/.append style={xshift=0.4mm, yshift=0.0mm},
	every y tick label/.append style={xshift=-0.6mm, yshift=0.0mm},
	every z tick label/.append style={xshift=-0.6mm, yshift=0.0mm},
	baseline,
	trim axis left, trim axis right,
	clip,
	every axis plot/.append style={},
	zmin=0,
	]
	\addplot3[surf, mesh/ordering=y varies, z buffer=none, mesh/cols=26, mesh/check=false,faceted color=black, point meta min=-0.5] table {./data/stable.dat};
	\end{axis}
	\end{tikzpicture}
	\caption{State profile $x(z,t)$ of the closed-loop PDE system \eqref{plant}, \eqref{sf} for the IC $x_0(z) = 3.1g(z)$.\label{fig:sprof}}
\end{figure}

\subsubsection{Feedback Linearization}
The bilinear system \eqref{nsys2} is obtained by choosing the estimated principal eigenvalues $\hat{\lambda}_i$, $i = 4,5$, and their eigenfunctionals to determine \eqref{nsys} for $n = 2$. Then, using the approach in Section \ref{sec:dewp} the matrices $b$ and $N$ follow from the data with $M = 200$, where the functional derivatives are approximated by utilizing multivariable Legendre polynomials. The resulting error norm $\|E_u[\underline{x}]\|_F = 0.004$ validates the choice of $M$. This results in the data-driven bilinear approximation model
\begin{equation}\label{blinappr}
\dot{w}(t) = \Lambda_2w(t) + (b + Nw(t))u(t)
\end{equation}
with
\begin{subequations}\label{blinmat}
\begin{align}
\Lambda_2 &= \begin{bsmallmatrix}
0.51769 & 0\\
0        & -9.1727
\end{bsmallmatrix}, \quad 
b = \begin{bsmallmatrix}
0.74563\\
0.037571
\end{bsmallmatrix}\\
N &= \begin{bsmallmatrix}
-1.4184 & 1.1191\\
0.026323 & -0.50604
\end{bsmallmatrix}.
\end{align}	 
\end{subequations}
In simulations it is verified that \eqref{blinmat} is a good approximation of \eqref{nsys} meaning that the state error $e_x[x]$ and the input error $e_u[x]$ are small if the deviation of $x$ from the origin is not too large. The data-driven state feedback controller \eqref{sf} is computed by solving the singular PDE \eqref{spde21} for \eqref{blinappr}. To this end, the matrix $\tilde{A} = \Lambda_2 - bk\t$ is assigned, which is obtained from the eigenvalue placement $\sigma(\tilde{A}) = \{-1,-12\}$ with the eigenvalues and the feedback gain $k\t$ satisfying all conditions of Theorem \ref{lem:singpde1}. Note that then all principal closed-loop eigenvalues are located in the open left half-plane. The corresponding linear target system reads
\begin{equation}\label{lintag}
 \dot{\tilde{w}}^{\text{lin}}(t) = \tilde{A}\tilde{w}^{\text{lin}}(t). 
\end{equation}
Using the Galerkin method in \cite{Deu08b} an approximate solution of \eqref{spde21} with minimal $L_2$-norm on a given interval $w \in I_w$ is found by employing multivariable Legendre polynomials. The approximation interval $I_w = [-\frac{13}{20},\frac{13}{20}]\times[-0.1,0.1]$ is chosen, in which the state $w(t)$, $t \geq 0$, of the closed-loop system \eqref{blinappr} and \eqref{sf} evolves for the considered IC $w(0) = \boldsymbol{\hat{\varphi}_2}[2.4g]$. Then, using the Galerkin approach with Legendre polynomials up to degree $n_{\text{leg}} = 11$ results in the $L_2$-norm $\|r_{n_{\text{leg}}}\|_2 = 3.7412\cdot10^{-14}$ of the approximation error for the solution of \eqref{spde21}. 

\subsubsection{Simulations}
Firstly, the feedback linearization of the closed-loop bilinear system \eqref{blinappr}, \eqref{sf} is investigated in a simulation. For this, the transformed closed-loop solution $\Phi(w)$ is compared with the solution $\tilde{w}^{\text{lin}}$ of the linear target system \eqref{lintag}. Figure \ref{fig:lin} verifies a good linearization result for the IC $w(0) = \boldsymbol{\hat{\varphi}_2}[2.4g]$ in the left column. By increasing the IC to $x_0(z) = 3.1g(z)$ the right column displays the closed-loop solution of the PDE system  \eqref{plant}, \eqref{sf} in the linearizing coordinates. The appearing deviations are attributed to the error $e_{\tilde{w}}(\tilde{w}(t),x(t))$, which is comparatively small. The evolution of the state profile $x(z,t)$ is depicted in Figure \ref{fig:sprof} and shows a good stabilization performance for the semilinear parabolic PDE. In simulations also a linear controller $u = - k\t \bm{\hat{\varphi}_2}[x]$ was tested, which also assigns the  eigenvalues $-1$ and $-12$ to the linearization of \eqref{blinappr}. This controller, however, was not able to stabilize the equilibrium $x = 0$ of the closed-loop system already for the smaller IC  $w(0) = \boldsymbol{\hat{\varphi}_2}[2.4g]$. By further increasing the IC, the closed-loop PDE system \eqref{plant}, \eqref{sf} also becomes unstable, i.e., only a local stabilization result is obtained. This is attributed to the utilized basis functionals, which are only cubic in the state. Additionally, it should be remarked that it is well-known that the system considered in the example is not globally stabilizable (see \cite{Fer00}).

\section{Concluding Remarks}
An obvious extension considers the data-driven design of observers. 
In order to further extend the applicability of the presented results, the consideration of quasilinear parabolic systems in the data-driven controller design is an interesting research topic. Another avenue for further research is to extend various methods for determining a suitable dictionary used in eDMD from finite to infinite dimensions.
\bibliographystyle{IEEEtranS}
\bibliography{IEEEabrv,mybib}

\end{document}